\documentclass[letterpaper]{article}

\usepackage[letterpaper,margin=1in]{geometry}
\usepackage{charter}
\usepackage{amsmath}
\usepackage{amssymb}
\usepackage{amsthm}
\usepackage{graphicx}
\usepackage[numbers]{natbib}
\allowdisplaybreaks[1]

\newcommand{\BigOh}[1]{O\!\left(#1\right)}

\newcommand{\LittleOmega}[1]{\omega\!\left(#1\right)}
\newcommand{\BigTheta}[1]{\Theta\!\left(#1\right)}

\newcommand{\layer}[1]{\mathsf{layer}[#1]}
\newcommand{\older}[1]{\mathsf{older}[#1]}
\newcommand{\younger}[1]{\mathsf{younger}[#1]}
\newcommand{\nextlayer}[1]{\mathsf{nextlayer}[#1]}
\newcommand{\aux}[1]{\mathsf{aux}[#1]}
\newcommand{\nil}{\mathsf{nil}}

\newcommand{\InsertFixUp}[1]{\textsc{Insert-FixUp}(#1)}
\newcommand{\DeleteFixUp}[1]{\textsc{Delete-FixUp}(#1)}
\newcommand{\Split}[1]{\textsc{Split}(#1)}
\renewcommand{\Join}[1]{\textsc{Join}(#1)}
\newcommand{\YoungestInLayer}[1]{\textsc{YoungestInLayer}(#1)}
\newcommand{\OldestInLayer}[1]{\textsc{OldestInLayer}(#1)}
\newcommand{\MoveUp}[1]{\textsc{MoveUp}(#1)}
\newcommand{\MoveDown}[1]{\textsc{MoveDown}(#1)}
\newcommand{\Search}[1]{\textsc{Search}(#1)}
\newcommand{\Insert}[1]{\textsc{Insert}(#1)}
\newcommand{\Delete}[1]{\textsc{Delete}(#1)}

\newtheorem{theorem}{Theorem}
\newtheorem{corollary}[theorem]{Corollary}
\newtheorem{lemma}[theorem]{Lemma}

\title{Layered Working-Set Trees}
\author{Prosenjit Bose\thanks{School of Computer Science, Carleton University. \texttt{\{jit,karim,vida,jhowat\}@cg.scs.carleton.ca}. This research was partially supported by NSERC and MRI.} \and Karim Dou\"ieb$^*$ \and Vida Dujmovi\'c$^*$ \and John Howat$^*$}
\date{}

\begin{document}
\maketitle

\begin{abstract}
The \emph{working-set bound} [Sleator and Tarjan, J. ACM, 1985] roughly states that searching for an element is fast if the element was accessed recently. Binary search trees, such as splay trees, can achieve this property in the amortized sense, while data structures that are not binary search trees are known to have this property in the worst case. We close this gap and present a binary search tree called a \emph{layered working-set tree} that guarantees the working-set property in the worst case. The \emph{unified bound} [B\u{a}doiu et al., TCS, 2007] roughly states that searching for an element is fast if it is near (in terms of rank distance) to a recently accessed element. We show how layered working-set trees can be used to achieve the unified bound to within a small additive term in the amortized sense while maintaining in the worst case an access time that is both logarithmic and within a small multiplicative factor of the working-set bound.
\end{abstract}

\section{Introduction}
\label{:introduction}

Let $S$ be a set of keys from a totally ordered universe and let $X$ be a sequence of elements from $S$. Typically, one is required to store elements of $S$ in some data structure $D$ such that accessing the elements of $S$ using $D$ in the order defined by $X$ is ``fast.'' Here, ``fast'' can be defined in many different ways, some focusing on worst case access times and others on amortized access times. For example, the search times of splay trees \cite{splay-trees} can be stated in terms of the rank difference between the current and previous elements of $X$; this is the \emph{dynamic finger property} \cite{dynamicfinger-2,dynamicfinger-1}.

If $x$ is the $i$-th element of $X$, we say that $x$ is accessed at time $i$ in $X$. The working-set number of $x$ at time $i$, denoted $w_i(x)$, is the number of distinct elements accessed since the last time $x$ was accessed or inserted, or $|D|$ if $x$ is either not in $D$ or has not been accessed by time $i$.

The \emph{working-set property} states the time to access $x$ at time $i$ is $\BigOh{\lg w_i(x)}$.\footnote{In this paper, $\lg x$ is defined to be $\log_2 (x+2)$.} Splay trees were shown by \citet{splay-trees} to have the working-set property in the amortized sense. One drawback of splay trees, however, is that most of the access bounds hold only in an amortized sense. While the amortized cost of a query can be stated in terms of its rank difference between successive queries or the number of distinct queries since a query was last made, any particular operation could take $\BigTheta{n}$ time. In order to address this situation, attention has turned to finding data structures that maintain the distribution-sensitive properties of splay trees but guarantee good performance in the worst case.

The data structure of \citet{unified}, called the working-set structure, guarantees this property in the worst case. However, this data structure departs from the binary search tree model and is instead a collection of binary search trees and queues.

\citet{unified} also describe a data structure called the unified structure that achieves the \emph{unified property}, which states that searching for $x$ at time $i$ takes time $\BigOh{\min_{y \in S} \lg (w_i(y) + d(x,y))}$ where $d(x,y)$ is the rank difference between $x$ and $y$. Again, this data structure is not a binary search tree. The skip-splay algorithm of \citet{skip-splay} fits into the binary search tree model and comes within a small additive term of the unified bound in an amortized sense.

\paragraph{Our Results.} We present a binary search tree that is capable of searching for a query $x$ in worst-case time $\BigOh{\lg w_i(x)}$ and performs insertions and deletions in worst-case time $\BigOh{\lg n}$, where $n$ is the number of keys stored by the tree at the time of the access. This fills in the gap between binary search trees that offer these query times in only an amortized sense and data structures which guarantee these query times in the worst-case but do not fit in the binary search tree model. We have also shown how to use this binary search tree to achieve the unified bound to within a small additive term in the amortized sense while maintaining in the worst case an access time that is both logarithmic and within a small multiplicative factor of the working-set bound.

\paragraph{Organization.} The rest of this paper is organized in the following way. We complete the introduction by summarizing the way the working-set structure of \citet{unified} operates, since this will play a key role in our binary search tree. In Section~\ref{:main}, we describe our binary search tree and explain the way in which operations are performed. In Section~\ref{:unified}, we show how to combine our results with those of \citet{skip-splay} on the unified bound to achieve an improved worst-case search cost. We conclude with Section~\ref{:conclusion} which summarizes our results and explains possible directions for future research.

\subsection{The Working-Set Structure}
\label{:introduction:wss}

We now describe the working-set structure of \citet{unified}. The structure maintains a dynamic set under the operations \textsc{Insert}, \textsc{Delete} and \textsc{Search}. Denote by $S_i \subseteq S$ the set of keys stored in the data structure at time $i$.

The structure is composed of $t = \BigOh{\lg \lg |S_i|}$ balanced binary search trees $T_1, T_2, \ldots, T_t$ and the same number of doubly linked lists $Q_1, Q_2, \ldots, Q_t$. For any $1 \le j \le t$, the contents of $T_j$ and $Q_j$ are identical, and pointers (in both directions) are maintained between their common elements. Every element in the set $S_i$ is contained in exactly one tree and in its corresponding list. For $j < t$, the size of $T_j$ and $Q_j$ is $2^{2^j}$, whereas the size of $T_t$ and $Q_t$ is $|S_i| - \sum_{j=1}^{t-1} 2^{2^j} \leq 2^{2^t}$. Figure~\ref{figure:ws-structure} shows a schematic of the structure.

\begin{figure}
   \begin{center}
      \includegraphics[width=0.6\columnwidth]{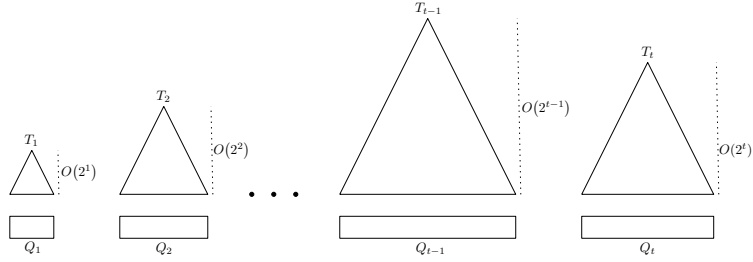}
   \end{center}
   \caption{The working-set structure of \citet{unified}. The pointers between corresponding elements in $T_j$ and $Q_j$ are not shown.}
   \label{figure:ws-structure}
\end{figure}

The working-set structure achieves its stated query time of $\BigOh{\lg w_i(x)}$ by ensuring that an element $x$ with working-set number $w_i(x)$ is stored in a tree $T_j$ with $j \leq \lceil \lg \lg w_i(x) \rceil$. Every list $Q_j$ orders the elements of $T_j$ by the time of their last access, starting with the youngest (most recently accessed) and ending with the oldest (least recently accessed).

Operations in the working-set structure are facilitated by an operation called a \emph{shift}. A shift is performed between two trees $T_j$ and $T_k$. Assume $j < k$, since the other case is symmetric. To perform a shift, we begin at $T_j$. We look in $Q_j$ to determine the oldest element and remove it from $Q_j$ and delete it from $T_j$. We then insert it into $T_{j+1}$ and $Q_{j+1}$ (as the youngest element) and repeat the process by shifting from $j+1$ to $k$. This process continues until we attempt to shift from one tree to itself. Observe that a shift causes the size of $T_j$ to decrease by one and the size of $T_k$ to increase by one. All of the trees between $T_j$ and $T_k$ will end up with the same size, but the elements contained in them change, since the oldest element from the previous tree is always added as the youngest element of the next tree.

We are now ready to describe how to make queries in the working-set structure. To search for an element $x$, we search sequentially in $T_1,T_2,\ldots$ until we find $x$ or search all of the trees and fail to find $x$. If $x \notin T_j$ for any $j$, then we will search every tree at a total cost of $\BigOh{\lg |S_i|}$ and then report that $x$ is not in the structure. Otherwise, assume $x \in T_j$. We delete $x$ from $T_j$ and $Q_j$ and insert it in $T_1$ and place it at the front of $Q_1$. We now have that the size of $T_1$ and $Q_1$ has increased by one and the size of $T_j$ and $Q_j$ has decreased by one. We therefore perform a shift from $1$ to $j$ to restore the sizes of the trees and lists. The time required for a search is dominated by the search time in $T_j$. Observe that if $x \in T_j$ and $j > 1$, then it must have been removed as the oldest element from $Q_{j-1}$, at which point at least $2^{2^{j-1}}$ distinct queries had been made. Therefore, $w_i(x) \ge 2^{2^{j-1}}$ and so the search time is $\BigOh{\lg 2^{2^j}} = \BigOh{\lg 2^{2^{j-1}}} = \BigOh{\lg w_i(x)}$.

Insertions are performed by inserting the element into $T_1$ and $Q_1$ (as the youngest element). Again, this causes $T_1$ and $Q_1$ to be too large. Since no other tree has space for one more element, we must shift to the last tree $T_t$. Thus, a shift from $1$ to $t$ is performed at total cost $\BigOh{\lg |S_i|}$. Note that it is possible that a new tree may need to be created if the size of $T_t$ grows past $2^{2^t}$. Deletions are performed by first searching for the element to be deleted. Once found, say in $T_j$, it is removed from $T_j$ and $Q_j$. To restore these sizes, we perform a shift from $t$ to $j$ at total cost $\BigOh{\lg |S_i|}$. If the last tree becomes empty, it can be removed.

\section{The Binary Search Tree}
\label{:main}

In this section, we describe a binary search tree that has the working-set property in the worst case.

\subsection{Model}
\label{:main:model}

Recall the binary search tree model of \citet{wilber-lowerbound}. Each node of the tree stores the key associated with it and has a pointer to its left and right children and its parent. The keys stored in the tree are from a totally ordered universe and are stored such that at any node, all of the keys in the left subtree are less than that stored in the node and all of the keys in the right subtree are greater than that stored at the node. Furthermore, each node may keep a constant\footnote{By standard convention, $\BigOh{\lg |S_i|}$ bits are considered to be ``constant.''} amount of additional information called \emph{fields}, but no additional pointers may be stored.

To perform an access to a key, we are given a pointer initialized to the root of the tree. An access consists of moving this pointer from a node to one of its adjacent nodes (through the parent pointer or one of the children pointers) until the pointer reaches the desired key. Along the way, we are allowed to update the fields and pointers in any nodes that the pointer reached. The access cost is the number of nodes reached by the pointer.

\subsection{Tree Decomposition}
\label{:main:decomposition}

Our binary search tree will adapt the working-set structure described in the previous section to the binary search tree model. Let $T$ denote the binary search tree as a whole. At a high level, our binary search tree layers the trees $T_1,T_2,\ldots,T_t$ of the working-set structure together to form $T$, and then augments nodes with enough information to recover which is the oldest in each tree at any given time.

Consider a labelling of $T$ where each node $x \in T$ has a label from $\{ 1,2,\ldots,t \}$ such that no node has an ancestor with a label greater than its own label. This labelling partitions the nodes of $T$. We say that the nodes with label $j \in \{ 1,2,\ldots,t \}$ form a \emph{layer} $L_j$. A layer $L_j$ will play the same role as $T_j$ in the working-set structure. Like $T_j$, $L_j$ contains exactly $2^{2^j}$ elements for $j < t$, and $L_t$ contains the remaining elements. Unlike $T_j$, $L_j$ is typically a collection of subtrees of $T$. We refer to a subtree of a layer $L_j$ as a \emph{layer-subtree}. Figure~\ref{figure:layers} shows this decomposition. Every node $x \in T$ stores as a field the value $j$ such that $x \in L_j$ which we denote by $\layer{x}$. We also record the total number of layers $t$ and the size of $L_t$ at the root as fields of each node.

\begin{figure}
   \begin{center}
      \includegraphics[width=0.2\columnwidth]{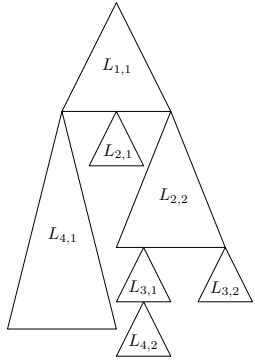}
   \end{center}
   \caption{The decomposition of the tree $T$ into layers. Here, the layer-subtrees of $L_j$ are denoted $L_{j,1},L_{j,2},\ldots$. Observe that layer $L_j$ can be connected to any layer $L_k$ with $k > j$. In this case, all of the elements of the layer-subtree $L_{4,1}$ are less than the elements in $L_{1,1}$, and so the layer-subtree $L_{4,1}$ must be connected to a leaf of $L_{1,1}$.}
   \label{figure:layers}
\end{figure}

Each layer-subtree $T_j' \in L_j$ is maintained independently as a tree that guarantees that each node of $T_j'$ has depth in $T_j'$ at most $\BigOh{\lg |T_j'|} = \BigOh{\lg |L_j|}$. This can be done using, \textit{e.g.}, a red-black tree \cite{symmetric-binary-btrees,redblack-trees}. By ``independently'', we mean that balance criteria are applied only to the elements within one layer-subtree. 

Our first observation concerns the depth of a node in a given layer.

\begin{lemma}
\label{lemma:depth}
The depth of a node $x \in L_j$ is $\BigOh{2^j}$.
\end{lemma}
\begin{proof}
In the worst case, we must traverse a layer-subtree of each of $L_1,L_2,\ldots,L_{j-1}$ to reach $L_j$ and then locate $x$ in $L_j$. Each layer $L_k$ has size $2^{2^k}$ and thus each layer-subtree we pass through has size at most $2^{2^k}$. Since each layer-subtree guarantees depth logarithmic in the size of the layer-subtree and thus the layer, the total depth is $\sum_{k=1}^{j} \BigOh{2^k} = \BigOh{2^j}$.
\end{proof}

The main obstacle in creating our tree comes from the fact that the core operations are performed on subtrees rather than trees, as is the case for the working-set structure. Consequently, standard red-black tree operations can not be used for the operations spanning more than one layer as described in Section~\ref{:main:inter-layer}. We break the operations into those restricted to one layer, those spanning two neighbouring layers, and finally those performed on the tree as a whole. These operations are described in the following sections.

Another difficulty arises from the having to implement the queues of the working-set structure in the binary search tree model. The queues are needed in order to determine the oldest element in a layer at any given time.

We encode the linked lists in our tree as follows. Each node $x \in L_j$ stores the key of the node inserted into $L_j$ directly before and after it. This information is stored in the fields $\older{x}$ and $\younger{x}$, respectively. We also store a key value in the field $\nextlayer{x}$. If $x$ is the oldest element in layer $L_j$, then no element was inserted before it and so we set $\older{x} = \nil$. In this case, we use $\nextlayer{x}$ to store the key of the oldest element in layer $L_{j+1}$. Similarly, if $x$ is the youngest element in layer $L_j$, then no element was inserted after it and so we set $\younger{x} = \nil$ and use $\nextlayer{x}$ to store the key of the youngest element in layer $L_{j+1}$. If $x$ is neither the youngest nor the oldest element in $L_j$, then we have $\nextlayer{x} = \nil$.

Before we describe how operations are performed on this binary search tree, we must make a brief note on storage. By the above description, each node $x$ stores three pointers (parent and children) and a key, as per the usual binary search tree model. The root also maintains the number of trees $t$ and the size of $L_t$. In addition, we must store balance information (one bit for red-black trees) and three additional key values (exactly one of which is $\nil$): $\older{x}$, $\younger{x}$ and $\nextlayer{x}$. If keys are assumed to be of size $\BigOh{\lg n}$, then it is clear our binary search tree fits the model of Section~\ref{:main:model}. Note that we are storing \emph{key values}, not pointers. Given a key value stored at a node, we do not have a pointer to it, so we must instead search for it  by traversing to the root and performing a standard search in a binary search tree. If keys have size $\LittleOmega{\lg n}$, it is true that we use more than $\BigOh{\lg n}$ additional space per node. However, since any node would then store a key of size $\LittleOmega{\lg n}$, we are only increasing the size of a node by a constant factor.

\subsection{Intra-Layer Operations}
\label{:main:intra-layer}

The operations we perform within a single layer are essentially the same as those we perform on any balanced binary search tree. We need notions of restoring balance after insertions and deletions and of splitting and joining. As mentioned before, we are not necessarily restricting ourselves to using any particular implementation of layer-subtrees. Instead, we will state the intra-layer operations and the required time bounds, and then show how red-black trees \cite{symmetric-binary-btrees,redblack-trees} can be used to fulfill this role. Other binary search trees that meet the requirements of each operation could also be used. Layer-subtrees must also ensure that their operations do not leave the layer-subtree; this can be done by checking the layer number of a node before visiting it.

Intra-layer operations rearrange layer-subtrees in some way. Observe that layer-subtrees hanging off a given node are maintained even after rearranging the layer-subtree, since the roots of such layer-subtrees can be viewed as the results of unsuccessful searches. Therefore, when describing these operations, we need not concern ourselves with explicitly maintaining layer-subtrees below the current one.

In our binary tree $T$, for each node $x$ in a layer-subtree $T_j'$ of $L_j$, we define the following operations. They are straightforward, but mentioned here for completeness and as a basis for the operations performed between layers.

\paragraph{$\InsertFixUp{x}$} This operation is responsible for ensuring that each node of $T_j'$ has depth $\BigOh{\lg |T_j'|}$ after the node $x$ has been inserted into the layer-subtree. For red-black trees, this operation is precisely the \textsc{RB-Insert-Fixup} operation presented by \citet[Section 13.3]{clrs}. Although the version presented there does not handle colouring $x$, it is straightforward to modify it to do so.

\paragraph{$\DeleteFixUp{x}$} This operation is responsible for ensuring that each node of $T_j'$ has depth $\BigOh{\lg |T_j'|}$ after a deletion in the layer-subtree. The exact node $x$ given to the operation is implementation dependent. For red-black trees, this operation is precisely the \textsc{RB-Delete-Fixup} operation presented by \citet[Section 13.4]{clrs}. In this case, the node $x$ is the child of the node spliced out by the deletion algorithm; we will elaborate on this when describing the layer operations in Section~\ref{:main:inter-layer}.

\paragraph{$\Split{x}$} This operation will cause the node $x \in T_j'$ to be moved to the root of $T_j'$. The rest of the layer-subtree will be split between the left and right side of $x$ such that each side is independently balanced and thus guarantee depth $\BigOh{\lg |T_j'|}$ of their respective nodes; this may mean that the layer-subtree is no longer balanced as a whole. For red-black trees, this operation is described by \citet[Chapter 4]{data-structures}, except we do not destroy the original trees, but rather stop when $x$ is the root of the layer-subtree.

\paragraph{$\Join{x}$} This operation is the inverse of $\Split{x}$: given a node $x \in T_j'$, we will restructure $T_j'$ to consist of $x$ at the root of the $T_j'$ and the remaining elements in subtrees rooted at the children of $x$ such that all nodes in the layer-subtree have depth $\BigOh{\lg |T_j'|}$. For red-black trees, this operation is described by \citet[Problem 13-2]{clrs}.

\begin{lemma}
\label{lemma:layer-subtree}
The operations $\InsertFixUp{x}$, $\DeleteFixUp{x}$, $\Split{x}$ and $\Join{x}$ on a node $x \in L_j$ can be implemented to take worst-case time $\BigOh{2^j}$ when red-black trees are used as layer-subtrees.
\end{lemma}
\begin{proof}
Immediate from the operations given by \citet{clrs} and \citet{data-structures}.
\end{proof}

\subsection{Inter-Layer Operations}
\label{:main:inter-layer}

The operations performed on layers correspond to the queue and shift operations of the working-set structure. The four operations performed on layers are $\YoungestInLayer{L_j}$ and $\OldestInLayer{L_j}$ for a layer $L_j$ and $\MoveUp{x}$ and $\MoveDown{x}$ for a node $x$.

As we did with the intra-layer operations, we will describe the requirements of the operations independently of the actual layer-subtree implementation. In fact, only the operation $\MoveDown{x}$ will require knowledge of the implementation of the layer-subtrees; the remaining operations simply make use of the operations defined in Section~\ref{:main:intra-layer}.

\paragraph{$\YoungestInLayer{L_j}$} This operation returns the key of the youngest node in layer $L_j$. We first examine all elements in $L_1$ (of which there are $\BigOh{1}$). Once we find the element that is the youngest (by looking for the element for which $\younger{x} = \nil$), say $x_1$, we go back to the root and search for $\nextlayer{x_1}$, which will bring us to the youngest element in $L_2$, say $x_2$. We then go back to the root and search for $\nextlayer{x_2}$, and so on. This repeats until we find the youngest element in $L_j$, as desired. The process for $\OldestInLayer{L_j}$ is the same, except our initial search in $L_1$ is for the oldest element, \textit{i.e.}, the element for which $\older{x} = \nil$.

\paragraph{$\MoveUp{x}$} This operation will move $x$ from its current layer $L_j$ to the next higher layer $L_{j-1}$. To accomplish this, we first split $x$ to the root of its layer-subtree using $\Split{x}$. We remove $x$ from $L_j$ by setting $\layer{x} = j - 1$. We now must restore balance properties. Observe that, by the definition of split, both of the layer-subtrees rooted at the children of $x$ are balanced. Therefore, we only need to ensure the balance properties $L_{j-1}$. Since we have just inserted $x$ into the layer $L_{j-1}$, this can be done by performing the intra-layer operation $\InsertFixUp{x}$. Finally, we must remove $x$ from the implicit queue structure of $L_j$ and place it in the implicit queue structure of $L_{j-1}$.

To do this, we look at both $\older{x}$ and $\younger{x}$. If they are both non-$\nil$, then we go to the root and perform searches for $\older{x}$ and $\younger{x}$, setting $\younger{\older{x}} = \younger{x}$ and $\older{\younger{x}} = \older{x}$. Otherwise, if only $\younger{x}$ is $\nil$, then we conclude that $x$ is the youngest in its former layer. After removing it from that layer, $\older{x}$ will be the new youngest element in that layer, so we go to the root search for $\older{x}$ and set $\younger{\older{x}} = \nil$. Since $\older{x}$ is the youngest element in that layer, we also copy $\nextlayer{x}$ into $\nextlayer{\older{x}}$. We must also update the key stored by the youngest element in the next higher layer. In order to do this, we run $\YoungestInLayer{L_{j-1}}$ to find this element, say $y$, and set $\nextlayer{y} = \older{x}$. The case for when only $\older{x}$ is $\nil$ is symmetric: the new oldest element in the layer is $\younger{x}$, so we update $\older{\younger{x}} = \nil$, we copy $\nextlayer{x}$ into $\nextlayer{\younger{x}}$, and update the pointer to the oldest element in this layer that is stored in $L_{j-1}$ in the same was as we did for the youngest.

We now must insert $x$ into the implicit queue structure of layer $L_{j-1}$. To do this, we search for the youngest node in $L_{j-1}$, say $y$. We then set $\older{x} = y$, $\younger{x} = \nil$ and $\younger{y} = x$. We then go to the next layer $L_{j-2}$ and update its pointer to the youngest element in this layer the same way we did before.

\paragraph{$\MoveDown{x}$} This operation will move $x$ from its current layer $L_j$ to the next lower layer $L_{j+1}$. We describe how to perform this operation for red-black trees; other implementations of the layer-subtrees will need to define different implementations but must respect the stated worst-case time bound of $\BigOh{2^j}$. Let $p$ denote the predecessor of $x$ in $L_j$. If $x$ does not have a predecessor in $L_j$, set $p = x$. Similarly, let $s$ denote the successor of $x$ in $L_j$, and if $x$ does not have a successor in $L_j$, set $s = x$. Our first goal is to move $x$ such that it becomes a leaf of its layer-subtree. If $x$ is not already a leaf in $L_j$, then $x$ has at least one child in its layer-subtree. To make it a leaf of it layer-subtree, we \emph{splice} out the node $s$ by making the parent of $s$ point to the right child of $s$ instead of $s$ itself. Note that this is well-defined since $s$ has no left child in $L_j$ as it is the smallest element greater than $x$. We then move $s$ to the location of $x$. Finally, we make $x$ a child of $p$ and make the new children of $x$ the old children of $p$ and $s$. Figure~\ref{figure:movedown} explains this process.

\begin{figure}
   \begin{center}
      \includegraphics[width=0.7\columnwidth]{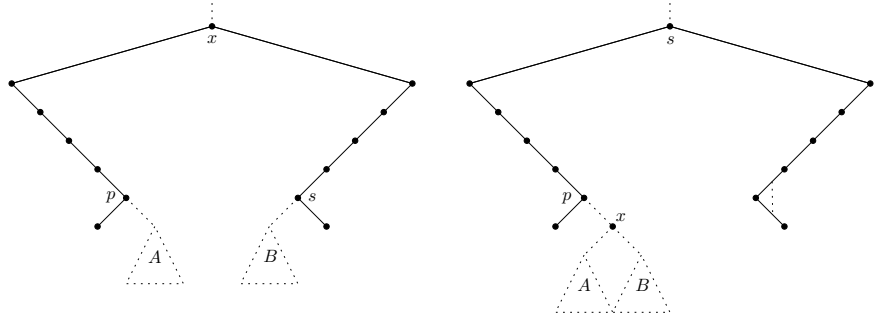}
   \end{center}
   \caption{The first part of the $\MoveDown{x}$ operation. On the left is the initial layer-subtree and the on the right is the layer-subtree after the nodes have been moved and layers changed but before the $\DeleteFixUp{s'}$. The dotted lines to nodes and subtrees indicate layer boundaries and the dotted line over the old node $s$ indicates a splice.}
   \label{figure:movedown}
\end{figure}

Observe that we now have that $x$ is a leaf of its layer-subtree. The layer-subtree is configured exactly as if we had deleted $x$ using the deletion operation described by \citet[Section 13.4]{clrs}. Therefore, we can perform $\DeleteFixUp{s'}$, where $s'$ is the (only) child of $s$, to restore the balance properties of the nodes of the layer-subtree. Thus, $s'$ is exactly the child of the node spliced out by the deletion ($s$), as required by the operation of \citet[Section 13.4]{clrs}.

To complete the movement to the next layer, we change the layer number of $x$ and execute $\Join{x}$ to create a single balanced layer-subtree from $x$ and its children.\footnote{Note that if these children have larger layer numbers than the new layer number for $x$, nothing is performed and $x$ becomes the lone element in its (new) layer-subtree; this follows from the fact that $\Join{x}$ only joins nodes that are in the same layer.} We then update the implicit queue structure as we did before. Observe that once $x$ has been removed from its original layer-subtree, layer-subtree balance has been restored because no node on that path was changed.

\begin{lemma}
\label{lemma:layer}
The operations $\YoungestInLayer{L_j}$ and $\OldestInLayer{L_j}$, $\MoveUp{x}$ and $\MoveDown{x}$ for a layer $L_j$ or a node $x \in L_j$ each take worst-case time $\BigOh{2^j}$.
\end{lemma}
\begin{proof}
The operations $\YoungestInLayer{L_j}$ and $\OldestInLayer{L_j}$ find the youngest (respectively oldest) element in layers $L_1,L_2,\ldots,L_j$. Given the youngest (respectively oldest) element in layer $L_k$, we can determine the youngest (respectively oldest) element in layer $L_{k+1}$ in constant time since such an element maintains the key of the youngest (respectively oldest) element in the next layer. We then need to traverse from the root to that element. By Lemma~\ref{lemma:depth}, the total time is $\sum_{k=1}^j \BigOh{2^k} = \BigOh{2^j}$.

The $\MoveUp{x}$ and $\MoveDown{x}$ operations, where $x \in L_j$, consist of searching for $x$, performing a constant number of intra-layer operations and then making series of queries for the youngest elements in several layers and updating the queue structures. The search can be done is $\BigOh{2^j}$ time by Lemma~\ref{lemma:depth} and the intra-layer operations each take $\BigOh{2^j}$ time by Lemma~\ref{lemma:layer-subtree} for a total of $BigOh{2^j}$. Finally, the queries for the youngest elements and the cost of updating the queues is dominated by the cost of the query in the deepest layer since each layer is twice the size of the previous one. Since $x \in L_j$, this cost is $\BigOh{2^j}$ by the above argument. The total cost of $\MoveUp{x}$ and $\MoveDown{x}$ is thus $\BigOh{2^j}$.
\end{proof}

\subsection{Tree Operations}
\label{:main:tree}

We are now ready to describe how to perform the operations $\Search{x}$, $\Insert{x}$ and $\Delete{x}$ on the tree as a whole. Such operations are independent of the layer-subtree implementation given the inter-layer and intra-layer operations defined in the previous sections.

\paragraph{$\Search{x}$} To perform a search for $x$, we begin by performing the usual method of searching in a binary search tree. Once we have found $x \in L_j$, we execute $\MoveUp{x}$ a total of $j-1$ times to bring $x$ into $L_1$. We then restore the sizes of the layers as was done in the working-set structure. We run $\OldestInLayer{L_1}$ to find the oldest element $y_1$ in layer $L_1$ and then run $\MoveDown{y_1}$. We then perform the same operation in $L_2$ by running $\OldestInLayer{L_2}$ to find the oldest element $y_2$ in layer $L_2$, then run $\MoveDown{y_2}$. This process of moving elements down layer-by-layer continues until we reach a layer $L_k$ such that $|L_k| < 2^{2^k}$.\footnote{Note that for an ordinary search, we have $k = j$. However, thinking of the algorithm this way gives us a clean way to describe insertions.} Note that efficiency can be improved by remembering the oldest elements of previous layers instead of finding the oldest element in each of $L_1,\ldots,L_j$ when running $\OldestInLayer{L_j}$. Such an improvement does not alter the asymptotic running time, however.

\paragraph{$\Insert{x}$} To insert $x$ into the tree, we first examine the index $t$ and size $|L_t|$ of the deepest layer, which we have stored at the root. If $|L_t| = 2^{2^t}$, then we increment $t$ and set $|L_t| = 1$. Otherwise, if $|L_t| < 2^{2^t}$, we simply increment $|L_t|$. We now insert $x$ into the tree (ignoring layers for now) using the usual algorithm where $x$ is placed in the tree as a leaf. We set $\layer{x} = t+1$ (\textit{i.e.}, a temporary layer larger than any other) and update the implicit queue structure for $L_t$ (and the youngest and oldest elements of $L_{t-1}$) as we did before. Finally, we run $\Search{x}$ to bring $x$ to $L_1$. Note that since $\Search{x}$ stops moving down elements once the first non-full layer is reached, we do not place another element in layer $t+1$. Thus, this layer is now empty and we update the youngest and oldest elements in layer $t$ to indicate that there is no layer below.

\paragraph{$\Delete{x}$} To delete $x$ from the tree, we look at the total number $t$ of layers in the tree that is stored at the root. We then locate $x \in T_j$ and perform $\MoveDown{x}$ a total of $t-j+1$ times. This will cause $x$ to be moved to a new (temporary) layer that is guaranteed to have no other nodes in it. Therefore, $x$ must be a leaf of the tree, and we can simply remove it by setting the corresponding child pointer of its parent to $\nil$. As was the case for insertion, this temporary layer is now empty and so we update the youngest and oldest elements in layer $t$ to indicate that there is no layer below. We then perform $t-j+1$ $\MoveUp{y}$ operations for the youngest element $y$ of each layer from $t$ to $j$ to restore the sizes of the layers. At this point, it could be the case that $|L_t| = 0$. If this happens, we decrement the number of layers $t$ which is stored at the root, and update the youngest and oldest elements in the new deepest layer to indicate that there is no layer below.

\begin{theorem}
\label{theorem:main}
Searching for $x$ at time $i$ takes worst-case time $\BigOh{\lg w_i(x)}$ and insertion and deletion each take worst-case time $\BigOh{\lg n}$.
\end{theorem}
\begin{proof}
A search consists of a regular search in a binary search tree followed by several layer operations. Suppose $x \in L_j$ at time $i$. By Lemma~\ref{lemma:depth}, we can find $x$ in time $\BigOh{2^j}$. We then perform $\MoveUp{x}$ in time $\BigOh{2^j}$ by Lemma~\ref{lemma:layer}. We then run $\textsc{OldestInLayer}$ and $\textsc{MoveDown}$ operations for every layer from $1$ to $j$. By Lemma~\ref{lemma:layer}, this has total cost $\sum_{k=1}^j \BigOh{2^k} = \BigOh{2^j}$. The total time is therefore $\BigOh{2^j}$. Observe that, by the same analysis as that of the working-set structure of \citet{unified}, we have that $w_i(x) \ge 2^{2^{j-1}}$, and so $\BigOh{2^j} = \BigOh{\lg w_i(x)}$.

An insertion consists of traversing through all layers. By Lemma~\ref{lemma:depth}, this takes time $\sum_{k=1}^t \BigOh{2^k} = \BigOh{2^t} = \BigOh{2^{\lg \lg n}} = \BigOh{\lg n}$. We then perform a search at cost $\BigOh{\lg n}$ by the above argument, since the element searched for is in the deepest layer. The total cost is thus $\BigOh{\lg n}$.

A deletion consists of traversing the tree to find $x \in L_j$ and then performing $\textsc{MoveDown}$ and $\textsc{MoveUp}$ at most once per layer. The traversal takes time $\BigOh{2^j}$ by Lemma~\ref{lemma:depth} and the $\textsc{MoveDown}$ and $\textsc{MoveUp}$ operations each cost $\BigOh{2^k}$ for $L_k$ by Lemma~\ref{lemma:layer}. The total cost is thus $\BigOh{2^j} + \sum_{k=1}^t \BigOh{2^k} = \BigOh{2^t} = \BigOh{2^{\lg \lg n}} = \BigOh{\lg n}$.
\end{proof}

\section{Skip-Splay and the Unified Bound}
\label{:unified}

In this section, we show how to use layered working-set trees in the skip-splay structure of \citet{skip-splay} in order to achieve the unified bound to within a small multiplicative factor. The unified bound \cite{unified} requires that the time to search an element $x$ at time $i$ is

\begin{displaymath}
\mathrm{UB}(x) = \BigOh{ \min_{y \in S_i} \lg (w_i(y) + d(x,y))}
\end{displaymath}

where $w_i(y)$ is the working-set number of $y$ at time $i$ (as in Section~\ref{:introduction}) and $d(x,y)$ is defined as the rank distance between $x$ and $y$. This property implies the working-set and the dynamic finger properties. Informally, the unified bound states that an access is fast if the current access is close in term of rank distance to some element that has been accessed recently. \citet{unified} introduced a data structure achieving the unified bound in the amortized sense. This structure does not fit into the binary search tree model, but the splay tree \cite{splay-trees}, which does fit into this model, is conjectured to achieve the unified bound \cite{unified}

Recently, \citet{skip-splay} developed the first binary search tree that guarantees an access time close to the unified bound. Their algorithm, called \emph{skip-splay}, performs an access to the element $x$ in $\BigOh{\mathrm{UB}(x) + \lg \lg n}$ amortized time. Insertions and deletions are not supported. In the remainder of this section, we briefly describe skip-splay and then show how to modify it using the layered working-set tree presented in Section~\ref{:main} in order to achieve a new bound in the binary search tree model.

The skip-splay algorithm works in the following way. Assume for simplicity that the tree $T$ stores the set $\{ 1,2,\ldots,n \}$ where $n = 2^{2^{k-1}} - 1$ for some integer $k \ge 0$ and that $T$ is initially perfectly balanced. Nodes of height $2^i$ (where the leaves of $T$ have height $1$) for $i \in \{0,1,\ldots,k-1\}$ are marked as the root of a subtree. Such nodes partition $T$ into a set of splay trees called \emph{auxiliary trees}. Each auxiliary tree is maintained as an independent splay tree. Observe that the $i$-th auxiliary tree encountered on a path from the root to a leaf in $T$ has size $2^{\lg_2 n/2^{i}} = n^{1/2^i}$. Define $\aux{x}$ to be the auxiliary tree containing the node $x$. 

To access an element $x$, we perform a standard binary search in $T$ to locate $x$. We then perform a series of splay operations on some of the auxiliary trees of $T$. We begin by splaying $x$ to the root of $\aux{x}$ using the usual splay algorithm. If $x$ is now the root of $T$, the operation is complete. Otherwise, we \emph{skip} to the new parent of $x$, say $y$, and splay $y$ to the root of $\aux{y}$. This process is repeated until we reach the root of $T$.

By using layered working-set trees as auxiliary trees in place of splay trees, we can get the following result.

\begin{theorem}
\label{theorem:unified1}
There exists a binary search tree that performs an access to the element $x_i$ in $\BigOh{\lg n}$ worst-case time and in $\BigOh{\mathrm{UB}(x_i) + \lg \lg n}$ amortized time.
\end{theorem}
\begin{proof}
As suggested by \citet{skip-splay}, instead of using splay trees to maintain the auxiliary trees, we could use any data structure that satisfies the working-set property. Thus, by maintaining the auxiliary trees as layered working-set, we straightforwardly guarantee an amortized time of $\BigOh{\mathrm{UB}(x_i) + \lg \lg n}$ to search for an element $x_i$. Note that the splay in the auxiliary tree corresponds to the $\Search{x}$ operation in our structure.

Now we show that this modified version of the skip-splay has the additional property that the worst case search time is $\BigOh{\lg n}$. A search consists of traversing a maximum of $k$ auxiliary trees where the size of the $i$-th encountered auxiliary tree is $n^{1/2^i}$. In the worst case, the amount of work performed in an auxiliary tree $A$ is $\BigOh{\lg |A|}$. Since the auxiliary trees are maintained independently from each other, the total worst-case search cost in the tree $T$ is $\BigOh{\sum^k_{i=1} \lg n/2^i} = \BigOh{\lg n}$.
\end{proof}

By doubling the access to an element, we also obtain the following result.

\begin{theorem}
\label{theorem:unified2}
The binary search tree described in Theorem~\ref{theorem:unified1} performs an access to the element $x_i$ in worst-case time $\BigOh{\lg \lg n \lg w_i(x_i)}$.
\end{theorem}
\begin{proof}
Doubling the access to an element increases by at most twice its worst-case access time. Thus, the asymptotic performance of the structure still holds for both the worst-case access time and amortized access time.

In order to reach an element in the tree, we have to traverse several auxiliary trees. Let $A_1,A_2,\ldots,A_k$ be the ordered sequence of trees traversed during an access to the element $x_i$ (note that $k \leq \lg \lg n$). The number of accesses performed independently in each of those trees is bounded above by $w_i(x_i)$.

For $j=1,2,\ldots,k-1$, define $d_i(A_j,A_{j+1})$ to be the distance between the root node of $A_j$ and the root node of $A_{j+1}$ in the structure at time $i$. More generally, define $d_i(A_j,y)$ as the distance between the root node of $A_j$ and the element $y$ where $y$ is a descendent of the root of $A_j$. Let $p(A_j)$ (and $s(A_j)$) be the greatest (smallest) element of $A_{j-1}$ that is smaller (greater) than any element in $A_j$. Thus the cost of accessing $x_i$ is $ \sum_{j=1}^{k-1} d_i(A_j,A_{j+1}) + d_i(A_k,x_i)$.

By the definition of a search tree we know that the parent of the root node of $A_j$ is either $p(A_j)$ or $s(A_j)$. Thus

\begin{equation}
\label{equation:depth}
d_i(A_j,A_{j+1})= \max \{ d_i(A_j,p(A_{j+1})), d_i(A_j,s(A_{j+1}))\} + 1.
\end{equation}

When we access $x_i$ twice, we independently access both $p(A_j)$ and $s(A_j)$ in each traversed auxiliary tree $A_j$. By Theorem \ref{theorem:main}, we have

\begin{equation*}
\left.
\begin{array}{r}
d_i(A_j,p(A_{j+1}))\\
d_i(A_j,s(A_{j+1}))
\end{array} \right\} = \BigOh{\lg w_i(x_i))} \quad {\rm for} \quad j=1,2,\ldots, k-1.
\end{equation*}

And we also have $d_i(A_k,x_i)= \BigOh{\lg w_i(x_i)}$. Hence, by applying equation (\ref{equation:depth}), the result follows.
\end{proof}

Note that this last property is not satisfied by the original unified structure \cite{unified}. Theorems \ref{theorem:unified1} and \ref{theorem:unified2} thus show

\begin{corollary}
\label{corollary:unified-main}
There exists a binary search tree that performs an access to the element $x_i$ in worst-case time $\BigOh{\min \{\lg n, (\lg \lg n)\lg w_i(x_i)\}}$ and in $\BigOh{\mathrm{UB}(x_i) + \lg \lg n}$ amortized time.
\end{corollary}

\section{Conclusion and Open Problems}
\label{:conclusion}

We have given the first binary search tree that guarantees the working-set property in the worst-case. We have also shown how to combine this binary search tree with the skip-splay algorithm of \citet{skip-splay} to achieve the unified bound to within a small additive term in the amortized sense while maintaining in the worst case an access time that is both logarithmic and within a small multiplicative factor of the working-set bound. Several directions remain for future research.

For layered working-set trees, it seems that by forcing the working-set property to hold in the worst case, we sacrifice good performance on some other access sequences. Is it the case that a binary search tree that has the working-set property in the worst case cannot achieve other properties of splay trees? For example, what kind of scanning bound can we achieve if we require the working-set property in the worst case? It would also be interesting to bound the number of rotations performed per access. Can we guarantee at most $\BigOh{\lg \lg w_i(x_i)}$ rotations to access $x_i$? Red-black trees guarantee $\BigOh{1}$ rotations per update, for instance.

For the results on the unified bound, the most obvious improvement would be to remove the $\lg \lg n$ term from the amortized access cost, as posed by \citet{skip-splay}. Another improvement would be to remove the $\lg \lg n$ factor from the worst-case access cost.

\paragraph{Acknowledgements.} We thank Jonathan Derryberry and Daniel Sleator for sending us a preliminary version of their skip-splay paper \cite{skip-splay} and Stefan Langerman for stimulating discussions.

\bibliography{references}

\begin{thebibliography}{10}
\providecommand{\natexlab}[1]{#1}
\providecommand{\url}[1]{\texttt{#1}}
\expandafter\ifx\csname urlstyle\endcsname\relax
  \providecommand{\doi}[1]{doi: #1}\else
  \providecommand{\doi}{doi: \begingroup \urlstyle{rm}\Url}\fi

\bibitem[Bayer(1972)]{symmetric-binary-btrees}
R.~Bayer.
\newblock Symmetric binary b-trees: Data structures and maintenance algorithms.
\newblock \emph{Acta Informatica}, 1:\penalty0 290--306, 1972.

\bibitem[B\u{a}doiu et~al.(2007)B\u{a}doiu, Cole, Demaine, and Iacono]{unified}
Mihai B\u{a}doiu, Richard Cole, Erik~D. Demaine, and John Iacono.
\newblock A unified access bound on comparison-based dynamic dictionaries.
\newblock \emph{Theoretical Computer Science}, 382\penalty0 (2):\penalty0
  86--96, 2007.

\bibitem[Cole(2000)]{dynamicfinger-2}
Richard Cole.
\newblock On the dynamic finger conjecture for splay trees. {Part II}: The
  proof.
\newblock \emph{SIAM J. Comput.}, 30\penalty0 (1):\penalty0 44--85, 2000.

\bibitem[Cole et~al.(2000)Cole, Mishra, Schmidt, and Siegel]{dynamicfinger-1}
Richard Cole, Bud Mishra, Jeanette Schmidt, and Alan Siegel.
\newblock On the dynamic finger conjecture for splay trees. {Part I}: Splay
  sorting $\log n$-block sequences.
\newblock \emph{SIAM J. Comput.}, 30\penalty0 (1):\penalty0 1--43, 2000.

\bibitem[Cormen et~al.(2001)Cormen, Leiserson, Rivest, and Stein]{clrs}
Thomas~H. Cormen, Charles~E. Leiserson, Ronald~L. Rivest, and Clifford Stein.
\newblock \emph{Introduction to Algorithms}.
\newblock MIT Press, 2nd edition, 2001.

\bibitem[Derryberry and Sleator(2009)]{skip-splay}
Jonathan~C. Derryberry and Daniel~D. Sleator.
\newblock Skip-splay: Toward achieving the unified bound in the {BST} model.
\newblock In \emph{{WADS '09}: Proceedings of the 16th Annual International
  Workshop on Algorithms and Data Structures}, 2009.

\bibitem[Guibas and Sedgewick(1978)]{redblack-trees}
Leonidas~J. Guibas and Robert Sedgewick.
\newblock A dichromatic framework for balanced trees.
\newblock In \emph{{FOCS '78: Proceedings of the 19th Annual IEEE Symposium on
  Foundations of Computer Science}}, pages 8--21, 1978.

\bibitem[Sleator and Tarjan(1985)]{splay-trees}
Daniel~Dominic Sleator and Robert~Endre Tarjan.
\newblock Self-adjusting binary search trees.
\newblock \emph{J. ACM}, 32\penalty0 (3):\penalty0 652--686, 1985.

\bibitem[Tarjan(1983)]{data-structures}
Robert~Endre Tarjan.
\newblock \emph{Data structures and network algorithms}.
\newblock Society for Industrial and Applied Mathematics, Philadelphia, PA,
  USA, 1983.

\bibitem[Wilber(1989)]{wilber-lowerbound}
Robert Wilber.
\newblock Lower bounds for accessing binary search trees with rotations.
\newblock \emph{SIAM Journal on Computing}, 18\penalty0 (1):\penalty0 56--67,
  1989.

\end{thebibliography}
\bibliographystyle{plainnat}
\end{document}